\documentclass{amsart}
\usepackage{graphicx}
\usepackage{amssymb}
\usepackage{amsfonts}
\usepackage{float}

\newtheorem{theorem}{Theorem}[section]

\theoremstyle{definition}

\theoremstyle{remark}

\numberwithin{equation}{section}
\theoremstyle{plain}

\newtheorem{notation}{Notation}

\newtheorem{proposition}{Proposition}

\copyrightinfo{2001}{enter name of copyright holder}

\begin{document}
\title[Turing patterns and $p$-adic FitzHugh-Nagumo systems]{Turing patterns
in a $p$-adic FitzHugh-Nagumo system on the unit ball}
\author[Chac\'{o}n-Cort\'{e}s]{L. F. Chac\'{o}n-Cort\'{e}s}
\address{Pontificia Universidad Javeriana, Departamento de Matem\'{a}ticas,
Cra. 7 N. 40-62, Bogot\'{a} D.C., Colombia}
\email{leonardo.chacon@javeriana.edu.co}
\author[Garcia-Bibiano]{C. A. Garcia-Bibiano}
\address{Centro de Investigaci\'{o}n y de Estudios Avanzados del Instituto
Polit\'{e}cnico Nacional. Departamento de Matem\'{a}ticas, Unidad Quer\'{e}%
taro. Libramiento Norponiente \#2000, Fracc. Real de Juriquilla. Santiago de
Quer\'{e}taro, Qro. 76230. M\'{e}xico}
\email{cagarcia@math.cinvestav.mx}
\author[Z\'{u}\~{n}iga-Galindo]{W. A. Z\'{u}\~{n}iga-Galindo$^{1}$}
\address{University of Texas Rio Grande Valley. School of Mathematical \&
Statistical Sciences. One West University Blvd. Brownsville, TX 78520,
United States}
\email{wilson.zunigagalindo@utrgv.edu}
\thanks{The third author was partially supported by the Lokenath Debnath
Endowed Professorship, UTRGV}

\begin{abstract}
We introduce discrete and $p$-adic continuous versions of the
FitzHugh-Nagumo system on $\mathbb{\ }$the one-dimensional $p$-adic unit
ball. We provide criteria for the existence of Turing patterns. We present
extensive simulations of some of these systems. The simulations show that
the Turing patterns are traveling waves in the $p$-adic unit ball.
\end{abstract}

\maketitle

\section{Introduction}

Several models involving parabolic equations have been used in neuroscience
for the propagation of nerve impulses. Among these models, the one of
FitzHugh-Nagumo plays a central role. Proposed in the 1950s by FitzHugh,
this model accurately explains the propagation of electric impulses along
the nerve axon of the giant squid, see \cite{Murray}, \cite{Perthame} and
the references \ therein. Nowadays FitzHugh-Nagumo system is the simplest
model to describe pulse propagation in a spatial region. The simplest
version of this system is 
\begin{equation}
\left\{ 
\begin{array}{l}
\partial _{t}u(x,t)=mu-u^{3}-v+L_{u}\Delta u \\ 
\\ 
\partial _{t}v(x,t)=c(u-av-b)+L_{v}\Delta v,%
\end{array}%
\right.  \label{FN_Classical}
\end{equation}%
where the system parameters $a,b,c,m,L_{u}$, and $L_{v}$, are assumed to be
positive, and the functions $u$ and $v$ depend on time $t\geq 0$ and the
position $x\in \mathbb{R}$ on the domain of interest. The variable $u$
promotes the self-growth of $u$ and, at the same time, the growth of $v$ and
can thus be named an activator, while $v$ plays the role of an inhibitor
that annuls the growth of $u$.

In this paper, we introduce a $p$-adic counterpart of the system (\ref%
{FN_Classical}). In the new model $x$ runs through the ring of $p$-adic
integers $\mathbb{Z}_{p}$, where $p$ is a fixed prime number, and $t$ is a
real variable. Geometrically, $\mathbb{Z}_{p}$ is an infinite rooted tree;
analytically, $\mathbb{Z}_{p}$ \ is a locally compact topological additive
group a very rich mathematical structure. The system takes the following form%
\begin{equation}
\left\{ 
\begin{array}{l}
{{\frac{\partial u}{\partial t}}(x,t)=f(u,v)-}\left( {\boldsymbol{D}%
_{0}^{\alpha }-\lambda }\right) {u(x,t)} \\ 
\\ 
{{\frac{\partial v}{\partial t}}(x,t)=g(u,v)-d\left( {\boldsymbol{D}%
_{0}^{\alpha }-\lambda }\right) v(x,t)}\text{,}\ \text{ }x\in \mathbb{Z}_{p}%
\text{, }\ t\geq 0,%
\end{array}%
\right.  \label{FN_p_adico}
\end{equation}%
where ${\boldsymbol{D}_{0}^{\alpha }-\lambda }$ is the Vladimirov operator
on $\mathbb{Z}_{p}$, and $f(u,v)=\mu u-u^{3}-v$, $g(u,v)=\gamma (u-\delta
v-\beta )$, where $\mu $, $\beta $ are real numbers, and $\gamma $, $\delta $%
, $d$ are positive real numbers.

This system admits a natural discretization of the form 
\begin{equation}
\left\{ 
\begin{array}{cc}
\frac{\partial }{\partial t}\left[ u_{L}(I,t)\right] _{I\in G_{L}}= & \left[
\mu u_{L}(I,t)-u_{L}^{3}(I,t)-v_{L}(I,t)\right] _{I\in G_{L}}-A_{L}^{\alpha }%
\left[ u_{L}(I,t)\right] _{I\in G_{L}} \\ 
&  \\ 
\frac{\partial }{\partial t}\left[ v_{L}(I,t)\right] _{I\in G_{L}}= & \left[
\gamma \left( u_{L}(I,t)-\delta v_{L}(I,t)-\beta \right) \right] _{I\in
G_{L}}-dA_{L}^{\alpha }\left[ v_{L}(I,t)\right] _{I\in G_{L}},%
\end{array}%
\right.  \label{FN_p_Discrete}
\end{equation}%
where $G_{L}$ is a finite rooted tree with $L$ levels, and matrix $%
A_{L}^{\alpha }$ is \ a discretization of operator ${\boldsymbol{D}%
_{0}^{\alpha }-\lambda }$.

We present Turing instability criteria for systems (\ref{FN_p_adico}) and (%
\ref{FN_p_Discrete}), see Theorems \ref{Theorem 1}, \ref{TheoremB}. The
conditions for the existence of Turing patterns for both systems are
essentially the same, except for one condition which involves a subset $%
\Gamma $ of the eigenvalues of ${\boldsymbol{D}_{0}^{\alpha }-\lambda }$, in
the case of system (\ref{FN_p_adico}), and a subset $\Gamma _{L}$ of the
eigenvalues of matrix $A_{L}^{\alpha }$, in the case of (\ref{FN_p_Discrete}%
). We provide extensive numerical simulations of some systems of type (\ref%
{FN_p_Discrete}); in particular, these experiments show that the Turing
patterns are traveling waves inside the unit ball $\mathbb{Z}_{p}$. Our
numerical experiments also show that the eigenvalues of matrix $%
A_{L}^{\alpha }$ approximate the eigenvalues of ${\boldsymbol{D}_{0}^{\alpha
}-\lambda }$. We conjecture that the Turing patterns of (\ref{FN_p_Discrete}%
) converge, in some sense, to the Turing patterns of (\ref{FN_p_adico}). The
results of Digernes and his collaborators on the problem of approximation of
spectra of Vladimirov \ operator $\boldsymbol{D}^{\alpha }$ by matrices of
type $A_{L}^{\alpha }$, \cite{Digernes}-\cite{Digernes et al} provide strong
support to our conjecture.

Nowadays, the study of Turing patterns on networks is a relevant area. In
the 70s, Othmer and Scriven pointed out that Turing instability can occur in
network-organized systems \cite{Othmer et al 1}-\cite{Othmer et al 2}. Since
then, reaction-diffusion models on networks have been studied intensively,
see, e.g., \cite{Ambrosio et al}, \cite{Boccaletti et al}, \cite{Chung}, 
\cite{Ide}, \cite{Mocarlo}- \cite{Mugnolo}, \cite{Nakao-Mikhailov}-\cite%
{Othmer et al 2}, \cite{Slavova et al}, \cite{von Below}-\cite{Zhao}, and
the references therein. In particular, Turing patterns of discrete
FitzHugh-Nagumo systems have also been studied \cite{Carletti}. In \cite%
{zuniga2020reaction}-\cite{W-morfo}, the last author established \ the
existence of Turing patterns for specific $p$-adic systems of
reaction-diffusion equations. Still, these papers need to consider the
problem of the numerical approximation of the Turing patterns.

The article is organized as follows. In Section \ref{Section_2}, we review
some basic aspects of the $p$-adic analysis and fix the notation. In Section %
\ref{Section_3} presents some basic aspects of the Vladimirov operator, the $%
p$-adic\ heat equation on the unit ball. In Section \ref{Section_4}, we
introduce our $p$-adic FitzHugh-Nagumo system, and give a Turing instability
criterion, see Theorem \ref{Theorem 1}. In Section \ref{Section_5}, we study
a discrete version of our $p$-adic FitzHugh-Nagumo system, and give a Turing
instability criterion, see Theorem \ref{TheoremB}. Finally, in Section \ref%
{Section_6}, we provide extensive numerical simulation for some discrete
FitzHugh-Nagumo systems and their Turing patterns.

\section{\label{Section_2}$p$-Adic Analysis: Essential Ideas}

In this section, we collect some basic results on $p$-adic analysis that we
use through the article. For a detailed exposition the reader may consult 
\cite{Alberio et al}, \cite{Kochubei}, \cite{Taibleson}, \cite{V-V-Z}.

\subsection{The field of $p$-adic numbers}

Along this article $p$ will denote a prime number. The field of $p$-adic
numbers $\mathbb{Q}_{p}$ is defined as the completion of the field of
rational numbers $\mathbb{Q}$ with respect to the $p$-adic norm $|\cdot|_{p}$%
, which is defined as 
\begin{equation*}
\left\vert x\right\vert _{p}=\left\{ 
\begin{array}{lll}
0 & \text{if} & x=0 \\ 
&  &  \\ 
p^{-\gamma} & \text{if} & x=p^{\gamma}\frac{a}{b}\text{,}%
\end{array}
\right.
\end{equation*}
where $a$ and $b$ are integers coprime with $p$. The integer $\gamma:=ord(x)$%
, with $ord(0):=+\infty$, is called the\textit{\ }$p$-\textit{adic order of} 
$x$.

Any $p$-adic number $x\neq0$ has a unique expansion of the form 
\begin{equation*}
x=p^{ord(x)}\sum_{j=0}^{\infty}x_{j}p^{j},
\end{equation*}
where $x_{j}\in\{0,\dots,p-1\}$ and $x_{0}\neq0$. By using this expansion,
we define \textit{the fractional part of }$x\in\mathbb{Q}_{p}$, denoted $%
\{x\}_{p}$, as the rational number 
\begin{equation*}
\left\{ x\right\} _{p}=\left\{ 
\begin{array}{lll}
0 & \text{if} & x=0\text{ or }ord(x)\geq0 \\ 
&  &  \\ 
p^{ord(x)}\sum_{j=0}^{-ord_{p}(x)-1}x_{j}p^{j} & \text{if} & ord(x)<0.%
\end{array}
\right.
\end{equation*}

\subsection{Basic topology of $\mathbb{Q}_{p}$}

For $r\in \mathbb{Z}$, denote by $B_{r}(a)=\{x\in \mathbb{Q}%
_{p};|x-a|_{p}\leq p^{r}\}$ \textit{the ball of radius }$p^{r}$ \textit{with
center at} $a\in \mathbb{Q}_{p}$, and take $B_{r}(0):=B_{r}$. We also denote
by $S_{r}(a)=\{x\in \mathbb{Q}_{p};|x-a|_{p}=p^{r}\}$ \textit{the sphere of
radius }$p^{r}$ \textit{with center at} $a\in \mathbb{Q}_{p}$, and take $%
S_{r}(0):=S_{r}$. We notice that $S_{0}^{1}=\mathbb{Z}_{p}^{\times }$ (the
group of units of $\mathbb{Z}_{p}$). The balls and spheres are both open and
closed subsets in $\mathbb{Q}_{p}$. In addition, two balls in $\mathbb{Q}%
_{p} $ are either disjoint or one is contained in the other. \newline
As a topological space $\left( \mathbb{Q}_{p},|\cdot |_{p}\right) $ is
totally disconnected, i.e., the only connected subsets of $\mathbb{Q}_{p}$
are the empty set and the points. A subset of $\mathbb{Q}_{p}$ is compact if
and only if it is closed and bounded in $\mathbb{Q}_{p}$, see e.g., \cite[%
Section 1.3]{V-V-Z}, or \cite[Section 1.8]{Alberio et al}. Since $(\mathbb{Q}%
_{p},+)$ is a locally compact topological group, there exists a Haar measure 
$dx$, which is invariant under translations, i.e., $d(x+a)=dx$. If we
normalize this measure by the condition $\int_{\mathbb{Z}_{p}}dx=1$, then $%
dx $ is unique.

\begin{notation}
We will use $\Omega\left( p^{-r}|x-a|_{p}\right) $ to denote the
characteristic function of the ball $B_{r}(a)$. For more general sets, we
will use the notation $1_{A}$ for the characteristic function of set $A$.
\end{notation}

\subsection{The Bruhat-Schwartz space}

A complex-valued function $\varphi$ defined on $\mathbb{Q}_{p}$ is \textit{%
called locally constant} if for any $x\in\mathbb{Q}_{p}$ there exist an
integer $l(x)\in\mathbb{Z}$ such that 
\begin{equation}
\varphi(x+x^{\prime})=\varphi(x)\text{ for any }x^{\prime}\in B_{l(x)}.
\label{local_constancy}
\end{equation}
A function $\varphi:\mathbb{Q}_{p}\rightarrow\mathbb{C}$ is called a \textit{%
Bruhat-Schwartz function (or a test function)} if it is locally constant
with compact support. Any test function can be represented as a linear
combination, with complex coefficients, of characteristic functions of
balls. The $\mathbb{C}$-vector space of Bruhat-Schwartz functions is denoted
by $\mathcal{D}(\mathbb{Q}_{p})$. We denote by $\mathcal{D}_{\mathbb{R}}(%
\mathbb{Q}_{p})$ the $\mathbb{R}$-vector space of Bruhat-Schwartz functions.
For $\varphi\in\mathcal{D}(\mathbb{Q}_{p})$, the largest number $%
l=l(\varphi) $ satisfying (\ref{local_constancy}) is called \textit{the
exponent of local constancy (or the parameter of constancy) of} $\varphi$.

We denote by $\mathcal{C}\left( \mathbb{Q}_{p}\right) $, the $\mathbb{C}$%
-vector space of continuous functions defined on $\mathbb{Q}_{p}$, and by $%
\mathcal{C}_{\mathbb{R}}\left( \mathbb{Q}_{p}\right) $ its real counterpart.

\subsection{$L^{\protect\rho}$ spaces}

Given an open subset $U$ of subset of $\mathbb{Q}_{p}$, and $\rho \in
\lbrack 1,\infty )$, we denote by $L^{\rho }\left( U\right) :=L^{\rho
}\left( U,dx\right) ,$ the $\mathbb{C}$-vector space of all the
complex-valued functions $\varphi $ satisfying 
\begin{equation*}
\left\Vert \varphi \right\Vert _{\rho }=\left\{ \int\limits_{U}\left\vert
\varphi \left( x\right) \right\vert ^{\rho }dx\right\} ^{\frac{1}{\rho }%
}<\infty .
\end{equation*}%
The corresponding $\mathbb{R}$-vector space are denoted as $L_{\mathbb{R}%
}^{\rho }\left( U\right) =L_{\mathbb{R}}^{\rho }\left( U,dx\right) $, $1\leq
\rho <\infty $. We denote by $\mathcal{D}(U)$ the $\mathbb{C}$-vector space
of test functions with supports contained in $U$, then $\mathcal{D}(U)$ is
dense in $L^{\rho }\left( U\right) $, for $1\leq \rho <\infty $, see e.g., 
\cite[Section 4.3]{Alberio et al}.

\subsection{The Fourier transform}

Set $\chi _{p}(y)=\exp (2\pi i\{y\}_{p})$ for $y\in \mathbb{Q}_{p}$. The map 
$\chi _{p}(\cdot )$ is an additive character on $\mathbb{Q}_{p}$, i.e., a
continuous map from $\left( \mathbb{Q}_{p},+\right) $ into $S$ (the unit
circle considered as multiplicative group) satisfying $\chi
_{p}(x_{0}+x_{1})=\chi _{p}(x_{0})\chi _{p}(x_{1})$, $x_{0},x_{1}\in \mathbb{%
Q}_{p}$. \ The additive characters of $\mathbb{Q}_{p}$ form an Abelian group
which is isomorphic to $\left( \mathbb{Q}_{p},+\right) $. The isomorphism is
given by $\kappa \rightarrow \chi _{p}(\kappa x)$, see, e.g., \cite[Section
2.3]{Alberio et al}.

The Fourier transform of $\varphi\in\mathcal{D}(\mathbb{Q}_{p})$ is defined
as 
\begin{equation*}
(\mathcal{F}\varphi(\xi)=\int \limits_{\mathbb{Q}_{p}}\chi_{p}(\xi
x)\varphi(x)dx\text{, for }\xi\in\mathbb{Q}_{p},
\end{equation*}
where $dx$ is the normalized Haar measure on $\mathbb{Q}_{p}$. We will also
use the notation $\mathcal{F}_{x\rightarrow\xi}\varphi$ and $\widehat{%
\varphi }$ for the Fourier transform of $\varphi$.

The Fourier transform extends to $L^{2}\left( \mathbb{Q}_{p}\right) $. If $%
f\in L^{2}\left( \mathbb{Q}_{p}\right) $, its Fourier transform is defined
as 
\begin{equation*}
(\mathcal{F}f)(\xi )=\lim_{k\rightarrow \infty }\int\limits_{|x|_{p}\leq
p^{k}}\chi _{p}(\xi x)f(x)dx,\quad \text{for }\xi \in 
\mathbb{Q}
_{p},
\end{equation*}%
where the limit is taken in $L^{2}\left( \mathbb{Q}_{p}\right) $.

\subsection{Distributions}

The $\mathbb{C}$-vector space $\mathcal{D}^{\prime }\left( \mathbb{Q}%
_{p}\right) $\ of all continuous linear functionals on $\mathcal{D}(\mathbb{Q%
}_{p})$ is called the \textit{Bruhat-Schwartz space of distributions}. Every
linear functional on $\mathcal{D}(\mathbb{Q}_{p})$ is continuous, i.e., $%
\mathcal{D}^{\prime }\left( \mathbb{Q}_{p}\right) $\ agrees with the
algebraic dual of $\mathcal{D}(\mathbb{Q}_{p})$, see e.g., \cite[Chapter 1,
VI.3, Lemma]{V-V-Z}. We denote by $\mathcal{D}_{\mathbb{R}}^{\prime }\left( 
\mathbb{Q}_{p}\right) $ the dual space of $\mathcal{D}_{\mathbb{R}}\left( 
\mathbb{Q}_{p}\right) $.

We endow $\mathcal{D}^{\prime}\left( \mathbb{Q}_{p}\right) $ with the weak
topology, i.e., a sequence $\left\{ T_{j}\right\} _{j\in\mathbb{N}}$ in $%
\mathcal{D}^{\prime}\left( \mathbb{Q}_{p}\right) $ converges to $T$ if $%
\lim_{j\rightarrow\infty}T_{j}\left( \varphi\right) =T\left( \varphi\right) $
for any $\varphi\in\mathcal{D}(\mathbb{Q}_{p})$. The map 
\begin{equation*}
\begin{array}{lll}
\mathcal{D}^{\prime}\left( \mathbb{Q}_{p}\right) \times\mathcal{D}(\mathbb{Q}%
_{p}) & \rightarrow & \mathbb{C} \\ 
\left( T,\varphi\right) & \rightarrow & T\left( \varphi\right)%
\end{array}%
\end{equation*}
is a bilinear form, which is continuous in $T$ and $\varphi$ separately. We
call this map the \textit{pairing} between $\mathcal{D}^{\prime}\left( 
\mathbb{Q}_{p}\right) $ and $\mathcal{D}(\mathbb{Q}_{p})$. From now on we
will use $\left( T,\varphi\right) $ instead of $T\left( \varphi\right) $.

Every $f$\ in $L_{loc}^{1}$ defines a distribution $f\in\mathcal{D}^{\prime
}\left( \mathbb{Q}_{p}\right) $ by the formula 
\begin{equation*}
\left( f,\varphi\right) =\int \limits_{\mathbb{Q}_{p}}f\left( x\right)
\varphi\left( x\right) dx.
\end{equation*}
Notice that for $f$\ $\in L_{\mathbb{R}}^{2}\left( \mathbb{Q}_{p}\right) $, $%
\left( f,\varphi\right) =\left\langle f,\varphi\right\rangle $, where $%
\left\langle \cdot,\cdot\right\rangle $ denotes the scalar product in $L_{%
\mathbb{R}}^{2}\left( \mathbb{Q}_{p}\right) $.

\subsection{The Fourier transform of a distribution}

The Fourier transform $\mathcal{F}\left[ T\right] $ of a distribution $T\in%
\mathcal{D}^{\prime}\left( \mathbb{Q}_{p}\right) $ is defined by 
\begin{equation*}
\left( \mathcal{F}\left[ T\right] ,\varphi\right) =\left( T,\mathcal{F}\left[
\varphi\right] \right) \text{ for all }\varphi\in\mathcal{D}(\mathbb{Q}_{p})%
\text{.}
\end{equation*}
The Fourier transform $T\rightarrow\mathcal{F}\left[ T\right] $ is a linear
and continuous isomorphism from $\mathcal{D}^{\prime}\left( \mathbb{Q}%
_{p}\right) $\ onto $\mathcal{D}^{\prime}\left( \mathbb{Q}_{p}\right) $.
Furthermore, $T=\mathcal{F}\left[ \mathcal{F}\left[ T\right] \left(
-\xi\right) \right] $.

\section{\label{Section_3} Vladimirov operator and $p$-adic wavelets}

\subsection{The $p$-adic heat equation}

For $\alpha >0$, the Vladimirov operator $\boldsymbol{D}^{\alpha }$ is
defined as 
\begin{equation*}
\begin{array}{ccc}
\mathcal{D}(\mathbb{Q}_{p}) & \rightarrow & L^{2}(\mathbb{Q}_{p})\cap 
\mathcal{C}\left( \mathbb{Q}_{p}\right) \\ 
&  &  \\ 
\varphi & \rightarrow & \boldsymbol{D}^{\alpha }\varphi ,%
\end{array}%
\end{equation*}%
where%
\begin{equation*}
\left( \boldsymbol{D}^{\alpha }\varphi \right) \left( x\right) =\frac{%
1-p^{\alpha }}{1-p^{-\alpha -1}}\int\limits_{\mathbb{Q}_{p}}\frac{\left[
\varphi \left( x-y\right) -\varphi \left( x\right) \right] }{\left\vert
y\right\vert _{p}^{\alpha +1}}dy.
\end{equation*}%
The $p$-adic analogue of the heat equation is%
\begin{equation*}
\frac{\partial u\left( x,t\right) }{\partial t}+a\boldsymbol{D}^{\alpha
}u\left( x,t\right) =0\text{, with }a>0\text{.}
\end{equation*}%
The solution of the Cauchy problem attached to the heat equation with
initial datum $u\left( x,0\right) =\varphi \left( x\right) \in \mathcal{D}(%
\mathbb{Q}_{p})$ is given by%
\begin{equation*}
u\left( x,t\right) =\int\limits_{\mathbb{Q}_{p}}Z\left( x-y,t\right)
\varphi \left( y\right) dy,
\end{equation*}%
where $Z\left( x,t\right) $ is the $p$\textit{-adic heat kernel} defined as 
\begin{equation}
Z\left( x,t\right) =\int\limits_{\mathbb{Q}_{p}}\chi _{p}\left( -x\xi
\right) e^{-at\left\vert \xi \right\vert _{p}^{\alpha }}d\xi ,
\label{het-kernel}
\end{equation}%
where $\chi _{p}\left( -x\xi \right) $\ is the standard additive character
of the group $\left( \mathbb{Q}_{p},+\right) $. The $p$-adic heat kernel is
the transition density function of a Markov stochastic process with space
state $\mathbb{Q}_{p}$, see, e.g., \cite{Kochubei}, \cite{Zuniga-LNM-2016}.

\subsection{The $p$-adic heat equation on the unit ball}

We define the operator $\boldsymbol{D}_{0}^{\alpha }$, $\alpha >0$, by
restricting $\boldsymbol{D}^{\alpha }$ to $\mathcal{D}(\mathbb{Z}_{p})$ and
considering $\left( \boldsymbol{D}^{\alpha }\varphi \right) \left( x\right) $
only for $x\in \mathbb{Z}_{p}$. The operator $\boldsymbol{D}_{0}^{\alpha }$\
satisfies 
\begin{equation}
\left( \boldsymbol{D}_{0}^{\alpha }-\lambda \right) \varphi (x)=\frac{%
1-p^{\alpha }}{1-p^{-\alpha -1}}\int\limits_{\mathbb{Z}_{p}}\frac{\varphi
(x-y)-\varphi (x)}{|y|_{p}^{\alpha +1}}dy,  \label{TV-Ball}
\end{equation}%
for $\mathbb{\varphi \in }\mathcal{D}(\mathbb{Z}_{p})$, with 
\begin{equation*}
\lambda =\frac{p-1}{p^{\alpha +1}-1}p^{\alpha }.
\end{equation*}

Consider the Cauchy problem%
\begin{equation*}
\left\{ 
\begin{array}{lll}
\frac{\partial u\left( x,t\right) }{\partial t}+\left( \boldsymbol{D}%
_{0}^{\alpha }-\lambda \right) u\left( x,t\right) =0\text{, } & x\in \mathbb{%
Z}_{p}, & t>0; \\ 
&  &  \\ 
u\left( x,0\right) =\varphi \left( x\right) , & x\in \mathbb{Z}_{p}, & 
\end{array}%
\right.
\end{equation*}%
where $\mathbb{\varphi \in }\mathcal{D}(\mathbb{Z}_{p})$. The solution of
this problem is given by%
\begin{equation*}
u\left( x,t\right) =\int\limits_{\mathbb{Z}_{p}}Z_{0}(x-y,t)\varphi \left(
y\right) dy\text{, }x\in \mathbb{Z}_{p}\text{, }t>0,
\end{equation*}%
where 
\begin{equation*}
Z_{0}(x,t):=e^{\lambda t}Z(x,t)+c(t)\text{, }x\in \mathbb{Z}_{p}\text{, }t>0,
\end{equation*}%
\begin{equation*}
c(t):=1-(1-p^{-1})e^{\lambda t}\sum_{n=0}^{\infty }\frac{(-1)^{n}}{n!}t^{n}%
\frac{1}{1-p^{-n\alpha -1}},
\end{equation*}%
and $Z(x,t)$ is given (\ref{het-kernel}). The function $Z_{0}(x,t)$ is
non-negative for $x\in \mathbb{Z}_{p}$, $t>0$, and 
\begin{equation*}
\int\limits_{\mathbb{Z}_{p}}Z_{0}(x,t)dx=1,
\end{equation*}%
\cite{Kochubei}. Furthermore, $Z_{0}(x,t)$ is the transition density
function of a Markov process with space state $\mathbb{Z}_{p}$.

\subsection{\label{Section_Spectrum}$p$-adic wavelets supported in balls}

The set of functions $\left\{ \Psi _{rnj}\right\} $ defined as%
\begin{equation}
\Psi _{rnj}\left( x\right) =p^{\frac{-r}{2}}\chi _{p}\left( p^{-1}j\left(
p^{r}x-n\right) \right) \Omega \left( \left\vert p^{r}x-n\right\vert
_{p}\right) ,  \label{eq4}
\end{equation}%
where $r\in \mathbb{Z}$, $j\in \left\{ 1,\cdots ,p-1\right\} $, and $n$ runs
through a fixed set of representatives of $\mathbb{Q}_{p}/\mathbb{Z}_{p}$,
is an orthonormal basis of $L^{2}(\mathbb{Q}_{p})$ consisting of
eigenvectors of operator $\boldsymbol{D}^{\alpha }:$ 
\begin{equation*}
\boldsymbol{D}^{\alpha }\Psi _{rnj}=p^{(1-r)\alpha }\Psi _{rnj}\text{ for
any }r,n,j,
\end{equation*}%
see, e.g., \cite[Theorem 3.29]{KKZuniga}, \cite[Theorem 9.4.2]{Alberio et al}%
. By using this basis, it is possible to construct an orthonormal basis for $%
L^{2}(\mathbb{Z}_{p})$:

\begin{proposition}[{\protect\cite[Propositions 1, 2]{Zuniga-EigenPardox}}]
The set of functions%
\begin{equation}
\left\{ \Omega\left( \left\vert x\right\vert _{p}\right) \right\}
\bigcup\bigcup\limits_{j\in\left\{ 1,\ldots,p-1\right\} }\text{ }%
\bigcup\limits_{r\leq0}\text{ }\bigcup\limits_{\substack{ np^{-r}\in \mathbb{%
Z}_{p}  \\ n\in\mathbb{Q}_{p}/\mathbb{Z}_{p}}}\left\{ \Psi_{rnj}\left(
x\right) \right\}  \label{eq4A}
\end{equation}
is an orthonormal basis of $L^{2}\left( \mathbb{Z}_{p}\right) $.
Furthermore, 
\begin{equation}
L^{2}(\mathbb{Z}_{p})=\mathbb{C}\Omega\left( \left\vert x\right\vert
_{p}\right) \bigoplus L_{0}^{2}(\mathbb{Z}_{p}),  \label{Resultado 2}
\end{equation}
where 
\begin{equation*}
L_{0}^{2}(\mathbb{Z}_{p})=\left\{ f\in L^{2}(\mathbb{Z}_{p});\int \limits_{%
\mathbb{Z}_{p}}f\text{ }dx=0\right\} .
\end{equation*}
\end{proposition}

Now, by using (\ref{TV-Ball}), (\ref{eq4}), (\ref{eq4A}), the functions in (%
\ref{eq4A}) are eigenfunctions of $\boldsymbol{D}_{0}^{\alpha }-\lambda $:%
\begin{equation}
\left( \boldsymbol{D}_{0}^{\alpha }-\lambda \right) \Psi _{rnj}=\left(
p^{(1-r)\alpha }-\lambda \right) \Psi _{rnj}\text{ }  \label{Resultado 3}
\end{equation}%
for any $r\leq 0$, $n\in p^{r}\mathbb{Z}_{p}\cap \mathbb{Q}_{p}/\mathbb{Z}%
_{p}$, $j\in \{1,\ldots ,p-1\}$, and 
\begin{equation}
\left( \boldsymbol{D}_{0}^{\alpha }-\lambda \right) \Omega \left( \left\vert
x\right\vert _{p}\right) =0\text{, for }x\in \mathbb{Z}_{p}.
\label{Resultado 4}
\end{equation}

\subsubsection{\label{Section_Eigenvalue_Problem} \textbf{An eigenvalue
problem in the unit ball}}

We now consider the following eigenvalue problem: 
\begin{equation}
\left\{ 
\begin{array}{ll}
\left( \boldsymbol{D}_{0}^{\alpha }-\lambda \right) \theta \left( x\right)
=\kappa \theta \left( x\right) \text{,} & \kappa \in \mathbb{R} \\ 
&  \\ 
\theta \in L_{\mathbb{R}}^{2}\left( \mathbb{Z}_{p}\right) . & 
\end{array}%
\right.  \label{E_Value_1}
\end{equation}%
By using (\ref{Resultado 3}), the functions $\Psi _{rnj}\left( x\right) $
given in (\ref{eq4A}) are complex-valued eigenfunctions of (\ref{E_Value_1})
with eigenvalues $\kappa \in \left\{ p^{\left( 1-r\right) \alpha }-\lambda
;r\leq 0\right\} $. Therefore, 
\begin{equation*}
\begin{split}
& p^{\frac{-r}{2}}\cos \left( \left\{ p^{r-1}jx-p^{-1}nj\right\} _{p}\right)
\Omega \left( \left\vert p^{r}x-n\right\vert _{p}\right) \text{, } \\
& p^{\frac{-r}{2}}\sin \left( \left\{ p^{r-1}jx-p^{-1}nj\right\} _{p}\right)
\Omega \left( \left\vert p^{r}x-n\right\vert _{p}\right) \text{,}
\end{split}%
\end{equation*}%
with $|p^{-r}n|_{p}\leq 1$ and $r\leq 0$, $n\in p^{r}\mathbb{Z}_{p}\cap 
\mathbb{Q}_{p}/\mathbb{Z}_{p}$, $j\in \{1,\ldots ,p-1\}$, are real-valued
eigenfunctions of (\ref{E_Value_1}) with $\kappa =p^{\left( 1-r\right)
\alpha }-\lambda $. By (\ref{Resultado 4}), $\Omega \left( \left\vert
x\right\vert _{p}\right) $ is a n eigenvalue for $\kappa =0$. Furthermore,
any $f(x)\in L_{\mathbb{R}}^{2}\left( \mathbb{Z}_{p}\right) $ admits an
expansion of the form 
\begin{align}
f(x)=& \sum\limits_{rnj}p^{\frac{-r}{2}}\text{Re}(A_{rnj})\cos \left(
\left\{ p^{r-1}jx-p^{-1}nj\right\} _{p}\right) \Omega \left( \left\vert
p^{r}x-n\right\vert _{p}\right)  \notag \\
& -\sum\limits_{rnj}p^{\frac{-r}{2}}\text{Im}(A_{rnj})\sin \left( \left\{
p^{r-1}jx-p^{-1}nj\right\} _{p}\right) \Omega \left( \left\vert
p^{r}x-n\right\vert _{p}\right)  \label{E_expansio} \\
& +A_{0}\Omega \left( \left\vert x\right\vert _{p}\right)  \notag
\end{align}%
where%
\begin{equation*}
\text{Re}(A_{rnj})=p^{\frac{-r}{2}}\int\limits_{\mathbb{Z}_{p}}f\left(
x\right) \cos \left( \left\{ p^{r-1}jx-p^{-1}nj\right\} _{p}\right) \Omega
\left( \left\vert p^{r}x-n\right\vert _{p}\right) dx,
\end{equation*}%
\begin{equation*}
\text{Im}(A_{rnj})=p^{\frac{-r}{2}}\int\limits_{\mathbb{Z}_{p}}f\left(
x\right) \sin \left( \left\{ p^{r-1}jx-p^{-1}nj\right\} _{p}\right) \Omega
\left( \left\vert p^{r}x-n\right\vert _{p}\right) dx\text{, }
\end{equation*}%
and 
\begin{equation*}
A_{0}=\int\limits_{\mathbb{Z}_{p}}f(x)dx.
\end{equation*}

\section{\label{Section_4} A $p$-adic FitzHugh-Nagumo system on $\mathbb{Z}%
_{p}$}

A reaction-diffusion system exhibits diffusion-driven instability, or Turing
instability, if the homogeneous steady state is stable to small
perturbations in the absence of diffusion but unstable to small spatial
perturbations when diffusion is present. The main process driving the
spatially inhomogeneous instability is diffusion: the mechanism determines
the spatial pattern that evolves. For Turing instability, we require that
the system is stable in the absence of diffusion.

From now on, we set $u(x,t),v(x,t):\mathbb{Z}_{p}\times \left[ 0,\infty
\right) \rightarrow \mathbb{R}$. We consider the following FitzHugh-Nagumo
system with $p$-adic diffusion: 
\begin{equation}
\left\{ 
\begin{array}{l}
u(\cdot ,t),v(\cdot ,t)\in L_{\mathbb{R}}^{2}(\mathbb{Z}_{p})\text{,}\ \text{%
\textup{for }}t\geq 0; \\ 
\\ 
u(x,0),v(x,0)\in L_{\mathbb{R}}^{2}(\mathbb{Z}_{p})\text{,}\ \ \
u(x,0),v(x,0)\text{, }x\in \mathbb{Z}_{p}; \\ 
\\ 
{{\frac{\partial u}{\partial t}}(x,t)=f(u,v)-}\left( {\boldsymbol{D}%
_{0}^{\alpha }-\lambda }\right) {u(x,t)}\text{,}\ \text{ \ }x\in \mathbb{Z}%
_{p}\text{, }t\geq 0; \\ 
\\ 
{{\frac{\partial v}{\partial t}}(x,t)=g(u,v)-d\left( {\boldsymbol{D}%
_{0}^{\alpha }-\lambda }\right) v(x,t)}\text{,}\ \ \text{ }x\in \mathbb{Z}%
_{p}\text{, }\ t\geq 0,%
\end{array}%
\right.  \label{1M}
\end{equation}%
where 
\begin{equation}
f(u,v)=\mu u-u^{3}-v\text{, }g(u,v)=\gamma (u-\delta v-\beta ),
\label{Definition_f_g}
\end{equation}%
and $\mu $, $\beta $, $\gamma \neq 0$, $\delta \neq 0$, $d$ are real numbers.

\subsection{Homogeneous steady states}

We now consider a homogeneous steady state (also called an equilibrium
point) of (\ref{1M}) which is a positive $(u_{0},v_{0})$ solution of 
\begin{equation}
\left\{ 
\begin{array}{l}
{\frac{\partial u}{\partial t}=f(u,v)},\quad t\geq 0 \\ 
\\ 
{\frac{\partial v}{\partial t}}=g(u,v),\quad t\geq 0.%
\end{array}%
\right.  \label{s1}
\end{equation}%
The equilibrium points associated with (\ref{s1}) are given by 
\begin{equation}
\left\{ 
\begin{array}{l}
\mu u-u^{3}-v=0 \\ 
\\ 
\gamma (u-\delta v-\beta )=0.%
\end{array}%
\right.  \label{s_h}
\end{equation}%
\newline
Using the substitution method on (\ref{s_h}), we have that (\ref{s_h}) is
equivalent to%
\begin{equation}
u^{3}+\eta u+\tau =0,  \label{s_h_1}
\end{equation}%
\newline
where $\eta :=\frac{1-\delta \mu }{\delta }$ and $\tau :=-\frac{\beta }{%
\delta }$. Here we use the hypothesis that $\gamma \neq 0$, $\delta \neq 0$%
.\ We denote by $u_{0}$ a real solution of (\ref{s_h_1}). Then $%
(u_{0},v_{0}) $, with $v_{0}=\frac{u_{0}-\beta }{\delta }$, is the real
equilibrium point of (\ref{s1}).

We denote by $\sigma _{\text{eigen}}\left( \boldsymbol{D}_{0}^{\alpha
}-\lambda \right) $ to the set of eigenvalues of $\boldsymbol{D}_{0}^{\alpha
}-\lambda $. We also set%
\begin{equation}
\kappa _{1}=\frac{1}{2d}\left\{ \left( d\left( \mu -3u_{0}^{2}\right)
-\gamma \delta \right) -\sqrt{\left( d\left( \mu -3u_{0}^{2}\right) -\gamma
\delta \right) ^{2}-4ddet (A)}\right\} ,  \label{k1}
\end{equation}%
\begin{equation}
\kappa _{2}=\frac{1}{2d}\left\{ \left( d\left( \mu -3u_{0}^{2}\right)
-\gamma \delta \right) +\sqrt{\left( d\left( \mu -3u_{0}^{2}\right) -\gamma
\delta \right) ^{2}-4ddet (A)}\right\} ,  \label{k2}
\end{equation}%
where 
\begin{equation*}
A=\left[ 
\begin{array}{ll}
f_{u} & f_{v} \\ 
g_{u} & g_{v}%
\end{array}%
\right] _{u=u_{0},v=v_{0}}:=\left[ 
\begin{array}{ll}
f_{u_{0}} & f_{v_{0}} \\ 
g_{u_{0}} & g_{v_{0}}%
\end{array}%
\right] .
\end{equation*}%
Notice that $A$ is the Jacobian matrix of the mapping $(u,v)\rightarrow
\left( f\left( u,v\right) ,g\left( u,v\right) \right) $. A straightforward
calculation shows that 
\begin{equation}
A=\left[ 
\begin{array}{cc}
\mu -3u_{0}^{2} & -1 \\ 
\gamma & -\gamma \delta%
\end{array}%
\right] .  \label{Matrix_A}
\end{equation}

\begin{theorem}
\label{Theorem 1}Consider the reaction-diffusion system (\ref{1M}). The
steady state $\left( u_{0},v_{0}\right) $ is linearly unstable (Turing
unstable), if the following conditions hold:

\begin{enumerate}
\item $Tr(A)=\mu -3u_{0}^{2}-\gamma \delta <0$ ;

\item $det (A)=-\mu \gamma \delta +3\gamma \delta u_{0}^{2}+\gamma >0$ ;

\item $d\left( \mu -3u_{0}^{2}\right) -\gamma \delta >0$;

\item The derivatives $\mu -3u_{0}^{2}$ and $-\gamma \delta $ must have
opposite signs;

\item $\left( d\left( \mu -3u_{0}^{2}\right) -\gamma \delta \right)
^{2}-4d\left( -\mu \gamma \delta +3\gamma \delta u_{0}^{2}+\gamma \right) >0$
;

\item $\Gamma =\{\kappa \in \sigma _{\text{eigen}}\left( \boldsymbol{D}%
_{0}^{\alpha }-\lambda \right) ;\kappa _{1}<\kappa <\kappa _{2}\}\neq
\emptyset $.
\end{enumerate}

Furthermore, there are infinitely many unstable eigenmodes, and the Turing
pattern has the form (\ref{17}).
\end{theorem}

\begin{proof}
The proof is similar to the one given in \cite{W-morfo}-\cite%
{zuniga2020reaction}. However, in \cite{W-morfo}, the Turing pattern is a
function from the subspace\ of $L_{\mathbb{R}}^{2}\left( \mathbb{Z}%
_{p}\right) $ consisting of functions with average zero, while here, the
pattern is a function from $L_{\mathbb{R}}^{2}\left( \mathbb{Z}_{p}\right) $%
. For further details the reader may consult the above mentioned references.
The Turing pattern $w(x,t)$ has the form%
\begin{align}
w(x,t)& \thicksim \sum_{\kappa _{1}<\kappa <\kappa
_{2}}\sum_{r,n}A_{rn}e^{\rho (\kappa )t}\Omega \left( \left\vert
p^{r}x-n\right\vert _{p}\right)  \label{17} \\
& +\sum_{\kappa _{1}<\kappa <\kappa _{2}}\sum_{r,n,j}A_{rnj}e^{\rho (\kappa
)t}p^{-\frac{r}{2}}\cos \left( \left\{ p^{-1}j\left( p^{r}x-n\right)
\right\} _{p}\right) \Omega \left( \left\vert p^{r}x-n\right\vert _{p}\right)
\notag \\
& +\sum_{\kappa _{1}<\kappa <\kappa _{2}}\sum_{r,n,j}B_{rnj}e^{\rho (\kappa
)t}p^{-\frac{r}{2}}\sin \left( \left\{ p^{-1}j\left( p^{r}x-n\right)
\right\} _{p}\right) \Omega \left( \left\vert p^{r}x-n\right\vert
_{p}\right) ,  \notag
\end{align}%
for $t\rightarrow \infty $, where $\rho (\kappa )$ are eigenvalues of matrix 
$A$ depending on $\kappa \in \sigma _{\text{eigen}}\left( \boldsymbol{D}%
_{0}^{\alpha }-\lambda \right) $, with $\text{Re}(\rho (\kappa
))>0 $.
\end{proof}

\section{\label{Section_5} Discrete FitzHugh-Nagumo systems}

\subsection{The Spaces $\mathcal{D}_{L}$}

We fix $L\in \mathbb{N\smallsetminus }\left\{ 0\right\} $, and define 
\begin{equation*}
G_{L}=\mathbb{Z}_{p}/p^{L}\mathbb{Z}_{p}.
\end{equation*}%
Then, $G_{L}$ is a finite ring, with $\#G_{L}=p^{L}$ elements. We set the
following set of representatives for the elements of $G_{L}$: 
\begin{equation*}
I=I_{0}+I_{q}p^{1}+\ldots +I_{L-1}p^{L-1},
\end{equation*}%
where the $I_{j}$s are $p$-adic digits, i.e., elements from $\{0,1,\ldots
,p-1\}$. We define $\mathcal{D}_{L}$ to be the space of test functions $%
\varphi $ supported in the unit ball having the form 
\begin{equation}
\varphi (x)=p^{\frac{L}{2}}\sum_{I\in G_{L}}\varphi (I)\Omega \left(
p^{L}|x-I|_{p}\right) \text{, with }\varphi (I)\in \mathbb{R}\text{. }
\label{D_{L}}
\end{equation}%
Since $\Omega \left( p^{L}|x-I|_{p}\right) \Omega \left(
p^{L}|x-J|_{p}\right) =0$ if $I\neq J$, the set 
\begin{equation*}
\left\{ p^{\frac{L}{2}}\Omega \left( p^{L}|x-I|_{p}\right) ;I\in
G_{L,M}\right\}
\end{equation*}%
is an orthonormal basis for $\mathcal{D}_{L}$. Then, by using that%
\begin{eqnarray*}
\Vert \varphi \Vert _{L^{2}} &=&\sqrt{p^{L}\sum_{I\in G_{L}}|\varphi
(I)|^{2}\int\limits_{\mathbb{Z}_{p}}\Omega \left( p^{L}|x-I|_{p}\right) dx}
\\
&=&\sqrt{\sum_{I\in G_{L}}|\varphi (I)|^{2}},
\end{eqnarray*}%
we have 
\begin{equation*}
\left( \mathcal{D}_{L},\Vert \cdot \Vert _{L^{2}}\right) \simeq \left( 
\mathbb{R}^{\#G_{L}},\left\vert \cdot \right\vert _{\mathbb{R}}\right) \text{
as Hilbert spaces, }
\end{equation*}%
where $\left\vert \cdot \right\vert _{\mathbb{R}}$ denotes the usual norm of 
$\mathbb{R}^{\#G_{L}}$.

\subsection{Discretization of the operator $\boldsymbol{D}_{0}^{\protect%
\alpha }-\protect\lambda $}

A natural discretization of $\boldsymbol{D}_{0}^{\alpha }-\lambda $ is
obtained by taking its restriction to $\mathcal{D}_{L}$. We denote this
restriction by $\boldsymbol{D}_{L}^{\alpha }-\lambda $. Since $\mathcal{D}%
_{L}$ is a finite vector space, $\boldsymbol{D}_{L}^{\alpha }-\lambda $ is
represented by a matrix $A_{L}^{\alpha }=\left[ A_{K,I}^{\alpha }\right]
_{K,I\in G_{L}}$, where 
\begin{equation}
A_{K,I}^{\alpha }=\left\{ 
\begin{array}{ll}
{p^{-\frac{L}{2}}\frac{1-p^{\alpha }}{1-p^{-\alpha -1}}\frac{1}{%
|K-I|_{p}^{\alpha +1}}} & \text{ if }\quad K\neq I; \\ 
&  \\ 
{-p^{-\frac{L}{2}}\frac{1-p^{\alpha }}{1-p^{-\alpha -1}}{\sum_{K\neq I}}%
\frac{1}{|K-I|_{p}^{\alpha +1}}} & \text{ if }\quad K=I,%
\end{array}%
\right.  \label{Discreto}
\end{equation}%
see \cite{W-morfo}. 

\begin{figure}[H]
\label{Figure 1}
\centering
\includegraphics[scale=0.4]{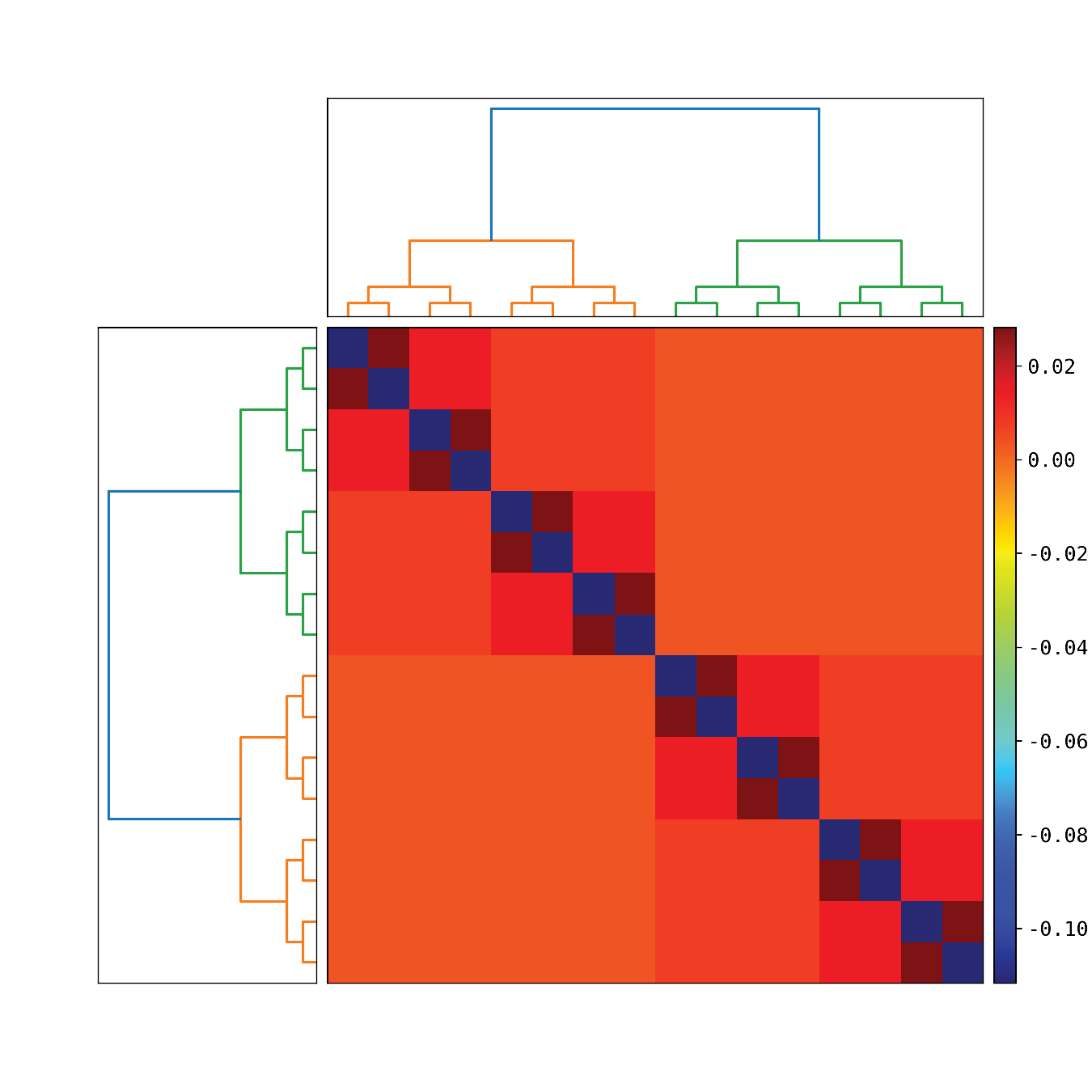}
\caption{The heat map for matrix $A_{L}^{\protect\alpha }$; $p=2,L=4,M=0,\protect\alpha =0.01$. The vertical and horizontal scales run through the points of tree $G_{4}$.}
\end{figure}

\subsection{Discretization of the $p$-adic Turing System (\protect\ref{1M})}

A discretization of the Turing system (\ref{1M}) is obtained by
approximating the functions $u(x,t)$, $v(x,t)$ as

\begin{equation*}
u_{L}(x,t)=\sum_{I\in G_{L}}u_{L}(I,t)\Omega \left( p^{L}|x-I|_{p}\right)
\end{equation*}%
and 
\begin{equation*}
v_{L}(x,t)=\sum_{I\in G_{L}}v_{L}(I,t)\Omega \left( p^{L}|x-I|_{p}\right) ,
\end{equation*}%
where $u_{L}(I,\cdot )$, $v_{L}(I,\cdot )\in C^{1}([0,T])$ for some fixed
positive $T$. We set%
\begin{equation*}
u_{L}(x,t)=\left[ u_{L}(I,t)\right] _{I\in G_{L}}\text{, }\ v_{L}(x,t)=\left[
v_{L}(I,t)\right] _{I\in G_{L}}.
\end{equation*}%
Notice that%
\begin{gather*}
f\left( \sum_{I\in G_{L}}u_{L}(I,t)\Omega \left( p^{L}|x-I|_{p}\right)
,\sum_{J\in G_{L}}u_{L}(J,t)\Omega \left( p^{L}|x-J|_{p}\right) \right) \\
=\sum_{I\in G_{L}}f\left( u_{L}(I,t),v_{L}(I,t)\right) \Omega \left(
p^{L}|x-I|_{p}\right) \\
=\sum_{I\in G_{L}}\left\{ \mu u_{L}(I,t)-u_{L}^{3}(I,t)-u_{L}(I,t)\right\}
\Omega \left( p^{L}|x-I|_{p}\right) .
\end{gather*}%
A similar formula holds for function $g$. Then, using (\ref{Definition_f_g}%
), the discretization of the $p$-adic Turing system (\ref{1M}) has the form:%
\begin{equation}
\left\{ 
\begin{array}{cc}
\frac{\partial }{\partial t}\left[ u_{L}(I,t)\right] _{I\in G_{L}}= & \left[
\mu u_{L}(I,t)-u_{L}^{3}(I,t)-v_{L}(I,t)\right] _{I\in G_{L}}-A_{L}^{\alpha }%
\left[ u_{L}(I,t)\right] _{I\in G_{L}} \\ 
&  \\ 
\frac{\partial }{\partial t}\left[ v_{L}(I,t)\right] _{I\in G_{L}}= & \left[
\gamma \left( u_{L}(I,t)-\delta v_{L}(I,t)-\beta \right) \right] _{I\in
G_{L}}-dA_{L}^{\alpha }\left[ v_{L}(I,t)\right] _{I\in G_{L}},%
\end{array}%
\right.  \label{Eq-TP_discrete}
\end{equation}%
where $A_{L}^{\alpha }=\left[ A_{K,I}^{\alpha }\right] _{K,I\in G_{L}}$.

We now rewrite system (\ref{Eq-TP_discrete}) in a matrix form. We denote by $%
diag\left( a_{I};I\in G_{L}\right) $, a diagonal matrix of size $%
\#G_{L}\times \#G_{L}$. Now, by using (\ref{Discreto}) and (\ref%
{Eq-TP_discrete}), we have%
\begin{gather}
\frac{\partial }{\partial t}\left[ 
\begin{array}{c}
\left[ u_{L}(I,t)\right] _{I\in G_{L}} \\ 
\\ 
\left[ v_{L}(I,t)\right] _{I\in G_{L}}%
\end{array}%
\right] =  \label{Eq-TP_discrete-2} \\
\notag \\
\left[ 
\begin{array}{ll}
diag\left( f\left( u_{L}(I,t),v_{L}(I,t)\right) ;I\in G_{L}\right) & {\LARGE %
0}_{\#G_{L}\times \#G_{L}} \\ 
&  \\ 
{\LARGE 0}_{\#G_{L}\times \#G_{L}} & diag\left( g\left(
u_{L}(I,t),v_{L}(I,t)\right) ;I\in G_{L}\right)%
\end{array}%
\right]  \notag \\
\notag \\
-\left[ 
\begin{array}{cc}
{\LARGE I}_{\#G_{L}\times \#G_{L}} & {\LARGE 0}_{\#G_{L}\times \#G_{L}} \\ 
&  \\ 
{\LARGE 0}_{\#G_{L}\times \#G_{L}} & d{\LARGE I}_{\#G_{L}\times \#G_{L}}%
\end{array}%
\right] \left[ 
\begin{array}{ll}
A_{L}^{\alpha } & {\LARGE 0}_{\#G_{L}\times \#G_{L}} \\ 
&  \\ 
{\LARGE 0}_{\#G_{L}\times \#G_{L}} & A_{L}^{\alpha }%
\end{array}%
\right] \left[ 
\begin{array}{c}
\left[ u_{L}(I,t)\right] _{I\in G_{L}} \\ 
\\ 
\left[ v_{L}(I,t)\right] _{I\in G_{L}}%
\end{array}%
\right] ,  \notag
\end{gather}%
where ${\LARGE 0}_{\#G_{L}\times \#G_{L}}$ denotes a matrix of size $%
\#G_{L}\times \#G_{L}$ with all its entries equal to zero, and ${\LARGE I}%
_{\#G_{L}\times \#G_{L}}$ denotes the identity matrix of size $\#G_{L}\times
\#G_{L}$.

\subsection{Discrete homogeneous steady states}

We study the equilibrium points of the system%
\begin{gather}
\frac{\partial }{\partial t}\left[ 
\begin{array}{c}
\left[ u_{L}(I,t)\right] _{I\in G_{L}} \\ 
\\ 
\left[ v_{L}(I,t)\right] _{I\in G_{L}}%
\end{array}%
\right] =  \label{System_1} \\
\left[ 
\begin{array}{ll}
diag\left( f\left( u_{L}(I,t),v_{L}(I,t)\right) ;I\in G_{L}\right) & {\LARGE %
0}_{\#G_{L}\times \#G_{L}} \\ 
&  \\ 
{\LARGE 0}_{\#G_{L}\times \#G_{L}} & diag\left( g\left(
u_{L}(I,t),v_{L}(I,t)\right) ;I\in G_{L}\right)%
\end{array}%
\right] .  \notag
\end{gather}
The equilibrium points are the solutions of the following system of
algebraic \ equations: 
\begin{equation}
\left\{ 
\begin{array}{l}
f\left( u_{L}(I),v_{L}(I)\right) =0 \\ 
\\ 
g\left( u_{L}(I),v_{L}(I)\right) =0,%
\end{array}%
\right.  \label{s_h_d}
\end{equation}%
where $I\in G_{L}$. Notice that if $f\left( u_{0},v_{0}\right) =g\left(
u_{0},v_{0}\right) =0$, then $u_{L}(I)=u_{0}$, $v_{L}(I)=v_{0}$ is a
solution of (\ref{s_h_d}) for any $I\in G_{L}$.

Take $\eta =\frac{1-\delta \mu }{\delta }$ and $\tau =-\frac{\beta }{\delta }
$, as before. Then, 
\begin{equation}
\left[ 
\begin{array}{c}
\left[ u_{0}\right] _{I\in G_{L}} \\ 
\\ 
\left[ v_{0}\right] _{I\in G_{L}}%
\end{array}%
\right]  \label{Equilibrium-point}
\end{equation}%
is one equilibrium point.

\subsection{The Jacobian matrix}

We now consider the following polynomial mapping:%
\begin{equation}
\begin{array}{ccc}
\mathbb{R}^{2\#G_{L}} & \rightarrow & \mathbb{R}^{2\#G_{L}} \\ 
&  &  \\ 
\left[ 
\begin{array}{c}
\left[ u_{L}(I)\right] _{I\in G_{L}} \\ 
\\ 
\left[ v_{L}(I)\right] _{I\in G_{L}}%
\end{array}%
\right] & \rightarrow & \left[ 
\begin{array}{c}
\left[ f\left( u_{L}(I),v_{L}(I)\right) \right] _{I\in G_{L}} \\ 
\\ 
\left[ g\left( u_{L}(I),v_{L}(I)\right) \right] _{I\in G_{L}}%
\end{array}%
\right] .%
\end{array}
\label{Mapping}
\end{equation}%
We denote by $\nabla f\left( u_{0},v_{0}\right) $, the $1\times 2$ matrix $%
\left[ 
\begin{array}{cc}
\frac{\partial f\left( u_{0},v_{0}\right) }{\partial u} & \frac{\partial
f\left( u_{0},v_{0}\right) }{\partial v}%
\end{array}%
\right] $, and by 
\begin{equation*}
diag\left( \nabla f\left( u_{0},v_{0}\right) ;I\in G_{L}\right) ,
\end{equation*}%
the block diagonal matrix 
\begin{equation*}
\left[ 
\begin{array}{ccc}
\nabla f\left( u_{0},v_{0}\right) &  & 0 \\ 
& \ddots &  \\ 
0 &  & \nabla f\left( u_{0},v_{0}\right)%
\end{array}%
\right]
\end{equation*}%
of \ size $\#G_{L}\times 2\#G_{L}$. In a similar form, we define the block
diagonal matrix 
\begin{equation*}
diag\left( \nabla g\left( u_{0},v_{0}\right) ;I\in G_{L}\right) .
\end{equation*}%
The Jacobian matrix $\mathcal{A}$ of mapping (\ref{Mapping}) at the
equilibrium point (\ref{Equilibrium-point}) is the $2\#G_{L}\times 2\#G_{L}$
matrix%
\begin{equation*}
\mathcal{A}=\left[ 
\begin{array}{ccc}
\nabla f\left( u_{0},v_{0}\right) &  & 0 \\ 
& \ddots &  \\ 
0 &  & \nabla f\left( u_{0},v_{0}\right) \\ 
\nabla g\left( u_{0},v_{0}\right) &  & 0 \\ 
& \ddots &  \\ 
0 &  & \nabla g\left( u_{0},v_{0}\right)%
\end{array}%
\right] =\left[ 
\begin{array}{c}
diag\left( \nabla f\left( u_{0},v_{0}\right) ;I\in G_{L}\right) \\ 
\\ 
diag\left( \nabla g\left( u_{0},v_{0}\right) ;I\in G_{L}\right)%
\end{array}%
\right] .
\end{equation*}%
We now set%
\begin{equation*}
A=\left[ 
\begin{array}{cc}
\frac{\partial f\left( u_{0},v_{0}\right) }{\partial u} & \frac{\partial
f\left( u_{0},v_{0}\right) }{\partial v} \\ 
\frac{\partial g\left( u_{0},v_{0}\right) }{\partial u} & \frac{\partial
g\left( u_{0},v_{0}\right) }{\partial v}%
\end{array}%
\right] =\left[ 
\begin{array}{c}
\nabla f\left( u_{0},v_{0}\right) \\ 
\nabla g\left( u_{0},v_{0}\right)%
\end{array}%
\right]
\end{equation*}%
as before, and by a finite sequence of swappings of rows, matrix $\mathcal{A}
$ can be written as%
\begin{equation}
\mathcal{A}^{\prime }=\left[ 
\begin{array}{ccc}
A &  & 0 \\ 
& \ddots &  \\ 
0 &  & A%
\end{array}%
\right] ,  \label{Matrix_A_Prima}
\end{equation}%
which is a $\#G_{L}\times \#G_{L}$ block matrix.

We denote by $\sigma \left( A_{L}^{\alpha }\right) $ the spectrum of $A$,
and use the $\kappa _{1}$, $\kappa _{2}$ defined in (\ref{k1})-(\ref{k2}).

\begin{theorem}
\label{TheoremB}Let us consider the reaction-diffusion system (\ref%
{Eq-TP_discrete-2}). The discrete steady state $\left[ 
\begin{array}{c}
\left[ u_{0}\right] _{I\in G_{L}} \\ 
\\ 
\left[ v_{0}\right] _{I\in G_{L}}%
\end{array}%
\right] $ is linearly unstable (Turing unstable), if the following
conditions hold:

\begin{enumerate}
\item $Tr(A)=\mu -3u_{0}^{2}-\gamma \delta <0$ ;

\item $det (A)=-\mu \gamma \delta +3\gamma \delta u_{0}^{2}+\gamma >0$ ;

\item $d\left( \mu -3u_{0}^{2}\right) -\gamma \delta >0$;

\item The derivatives $\mu -3u_{0}^{2}$ and $-\gamma \delta $ must have
opposite signs;

\item $\left( d\left( \mu -3u_{0}^{2}\right) -\gamma \delta \right)
^{2}-4d\left( -\mu \gamma \delta +3\gamma \delta u_{0}^{2}+\gamma \right) >0$
;

\item $\Gamma _{L}=\{\kappa _{L}\in \sigma \left( A_{L}^{\alpha }\right)
;\kappa _{1}<\kappa _{L}<\kappa _{2}\}\neq \emptyset $.
\end{enumerate}

Furthermore, the Turing pattern has the form (\ref{Turing_pattern}).
\end{theorem}

\begin{proof}
We first linearize system (\ref{Eq-TP_discrete-2}) about the steady state (%
\ref{Equilibrium-point}). Set%
\begin{equation*}
\left[ 
\begin{array}{c}
\left[ w_{L}^{\left( 1\right) }(I,t)\right] _{I\in G_{L}} \\ 
\\ 
\left[ w_{L}^{\left( 2\right) }(I,t)\right] _{I\in G_{L}}%
\end{array}%
\right] :=\left[ 
\begin{array}{c}
\left[ u_{L}(I,t)-u_{0}\right] _{I\in G_{L}} \\ 
\\ 
\left[ v_{L}(I,t)-v_{0}\right] _{I\in G_{L}}%
\end{array}%
\right] .
\end{equation*}%
Then the linear approximation is%
\begin{equation*}
\left[ 
\begin{array}{c}
\left[ w_{L}^{\left( 1\right) }(I,t)\right] _{I\in G_{L}} \\ 
\\ 
\left[ w_{L}^{\left( 2\right) }(I,t)\right] _{I\in G_{L}}%
\end{array}%
\right] =\mathcal{A}\left[ 
\begin{array}{c}
\left[ w_{L}^{\left( 1\right) }(I,t)\right] _{I\in G_{L}} \\ 
\\ 
\left[ w_{L}^{\left( 2\right) }(I,t)\right] _{I\in G_{L}}%
\end{array}%
\right] .
\end{equation*}%
The equilibrium point%
\begin{equation}
\left[ 
\begin{array}{c}
\left[ 0\right] _{I\in G_{L}} \\ 
\\ 
\left[ 0\right] _{I\in G_{L}}%
\end{array}%
\right]  \label{Equilibrium-point-2}
\end{equation}%
is linearly stable, if the eigenvalues of $\mathcal{A}$ have negative real
parts. By a suitable sequence of swappings of the rows of $\mathcal{A}$, we
have 
\begin{equation*}
det \left( \mathcal{A}-\rho I\right) =\pm det \left( \mathcal{A}^{\prime
}-\rho I\right) =\pm det \left( A-\rho I\right) ^{\#G_{L}}.
\end{equation*}%
Then the eigenvalues of $\mathcal{A}$ are exactly the eigenvalues of $A$
counted with multiplicity $G_{L}$:%
\begin{equation}
det \left( A-\rho I\right) =det \left[ 
\begin{array}{cc}
\mu -3u_{0}^{2}-\rho & -1 \\ 
\gamma & -\lambda \delta -\rho%
\end{array}%
\right] =\rho ^{2}-\rho Tr(A)+det (A)=0.  \label{Condition_1}
\end{equation}%
Then%
\begin{equation*}
\rho _{1,2}=\frac{\pm \sqrt{\left( \mu -3u_{0}^{2}-\gamma \delta \right)
^{2}-4\left( -\mu \gamma \delta +3\gamma \delta u_{0}^{2}+\gamma \right) }}{2%
}+\frac{\mu -3u_{0}^{2}-\gamma \delta }{2}.
\end{equation*}%
The condition $\text{Re}(\rho _{1,2})<0$ is guaranteed, if the trace and the
determinant of matrix $A$ satisfy%
\begin{equation}
\text{Tr}(A)<0\text{, }det (A)>0.  \label{Condition_det_tr}
\end{equation}%
Now, we linearize the entire reaction-ultradiffusion system close to the
steady state (\ref{Equilibrium-point-2}):%
\begin{equation}
\frac{\partial }{\partial t}\left[ 
\begin{array}{c}
\left[ w_{L}^{\left( 1\right) }(I,t)\right] _{I\in G_{L}} \\ 
\\ 
\left[ w_{L}^{\left( 2\right) }(I,t)\right] _{I\in G_{L}}%
\end{array}%
\right] =\left( \mathcal{A}-D_{L}\mathcal{A}_{L}^{\alpha }\right) \left[ 
\begin{array}{c}
\left[ w_{L}^{\left( 1\right) }(I,t)\right] _{I\in G_{L}} \\ 
\\ 
\left[ w_{L}^{\left( 2\right) }(I,t)\right] _{I\in G_{L}}%
\end{array}%
\right] ,  \label{System_4A}
\end{equation}%
where%
\begin{equation*}
D_{L}:=\left[ 
\begin{array}{cc}
{\LARGE I}_{\#G_{L}\times \#G_{L}} & {\LARGE 0}_{\#G_{L}\times \#G_{L}} \\ 
&  \\ 
{\LARGE 0}_{\#G_{L}\times \#G_{L}} & d{\LARGE I}_{\#G_{L}\times \#G_{L}}%
\end{array}%
\right] \text{, \ }\mathcal{A}_{L}^{\alpha }:=\left[ 
\begin{array}{ll}
A_{L}^{\alpha } & {\LARGE 0}_{\#G_{L}\times \#G_{L}} \\ 
&  \\ 
{\LARGE 0}_{\#G_{L}\times \#G_{L}} & A_{L}^{\alpha }%
\end{array}%
\right] .
\end{equation*}%
The matrices $A_{L}^{\alpha }$, $\mathcal{A}_{L}^{\alpha }$ are real
symmetric, and consequently they are diagonalizable. Then, there exists a
basis $\left\{ \boldsymbol{e}_{\kappa }\right\} $ of $\mathbb{R}^{\#G_{L}}$
such that 
\begin{equation*}
A_{L}^{\alpha }\boldsymbol{e}_{\kappa }=\kappa \boldsymbol{e}_{\kappa },
\end{equation*}%
where $\kappa =\kappa \left( L\right) $. Then%
\begin{equation*}
\mathcal{A}_{L}^{\alpha }\left[ 
\begin{array}{c}
\boldsymbol{e}_{\kappa } \\ 
\boldsymbol{e}_{\kappa }%
\end{array}%
\right] =\kappa \left[ 
\begin{array}{c}
\boldsymbol{e}_{\kappa } \\ 
\boldsymbol{e}_{\kappa }%
\end{array}%
\right] .
\end{equation*}%
We now look for a solution of system (\ref{System_4A}) of the form 
\begin{equation*}
\left[ 
\begin{array}{c}
\left[ w_{L}^{\left( 1\right) }(I,t)\right] _{I\in G_{L}} \\ 
\\ 
\left[ w_{L}^{\left( 2\right) }(I,t)\right] _{I\in G_{L}}%
\end{array}%
\right] ,
\end{equation*}%
where $w_{L}^{(j)}(I,t)=\sum_{\kappa ,\rho }C_{\kappa ,\rho }e^{\rho t}%
\boldsymbol{e}_{\kappa }$, where $\rho =\rho \left( j,I,L\right) $, $\kappa
=\kappa \left( j,I,L\right) $. The function $e^{\rho t}\boldsymbol{e}%
_{\kappa }$ is a non-trivial solution of (\ref{System_4A}), if $\rho $
satisfies 
\begin{equation}
\text{det }(\rho I-\mathcal{A}+\kappa D_{L})=0.  \label{6_d}
\end{equation}%
By a finite sequence of swappings of rows, we have%
\begin{gather}
\text{det }(\rho I-\mathcal{A}+\kappa D_{L})=\pm det \left[ 
\begin{array}{ccc}
\rho I_{2\times 2}-A+\kappa D &  & 0 \\ 
& \ddots &  \\ 
0 &  & \rho I_{2\times 2}-A+\kappa D%
\end{array}%
\right]  \notag \\
=\pm det \left( \rho I_{2\times 2}-A+\kappa D\right) ^{\#G_{L}}  \notag \\
=\rho ^{2}+[\kappa (1+d)-\text{Tr}(A)]\rho +h(\kappa )=0,
\label{Condition_2}
\end{gather}%
where 
\begin{equation*}
D=\left[ 
\begin{array}{ll}
1 & 0 \\ 
0 & d%
\end{array}%
\right] ,
\end{equation*}%
and 
\begin{equation}
h(\kappa ):=d\kappa ^{2}-\kappa \left( d\left( \mu -3u_{0}^{2}\right)
-\gamma \delta \right) +\text{det }(A).  \label{Eq_h_k}
\end{equation}%
Since $\kappa =0$ is not an eigenvalue of the matrix $A_{L}^{\alpha }$, the
conditions (\ref{Condition_1}) and (\ref{Condition_2}) are independent. \
For that the steady state to be unstable for spatial perturbations, we need
that $\text{Re}(\rho (\kappa ))>0$, for some $\kappa \neq 0$, this
can happen either if the coefficient of $\rho $ in (\ref{Condition_2}) is
negative or if $h(\kappa )<0$, for some $\kappa \neq 0$ in (\ref{Eq_h_k}).
For being $\text{Tr}(A)<0$ of the conditions (\ref{Condition_det_tr}) and
the coefficient of $\rho $ in (\ref{Condition_2}) is $\kappa (1+d)-\text{%
Tr}(A)$, which is positive, so the only way that $\text{\text{Re}}(\rho
(\kappa ))$ can be positive is if $h(\kappa )<0$ for some $\kappa \neq 0$.
As $\text{det }(A)>0$ of (\ref{Condition_det_tr}), in order for $%
h(\kappa )$ to be negative, it is necessary that $d\left( \mu
-3u_{0}^{2}\right) -\gamma \delta >0$. Now, since $\text{Tr}(A)=\mu
-3u_{0}^{2}-\gamma \delta <0$, necessarily $d\neq 1$ and $\mu -3u_{0}^{2}$
and $-\gamma \delta $ must have opposite signs. Thus, we have that an
additional requirement to (\ref{Condition_det_tr}) is that $d\neq 1$. This
is a necessary, but not sufficient, condition for that $\text{\text{Re}}%
(\rho (\kappa ))>0$. For that $h(\kappa )$ to be negative for some non zero $%
\kappa $, the minimum $h_{\text{min}}$ of $h(\kappa )$ must be negative.
Using elementary calculations, we show that%
\begin{equation*}
h_{\text{min}}=\text{det }(A)-\frac{\left( d\left( \mu
-3u_{0}^{2}\right) -\gamma \delta \right) ^{2}}{4d},
\end{equation*}%
and the minimum is reached at 
\begin{equation}
k_{\text{min}}=\frac{d\left( \mu -3u_{0}^{2}\right) -\gamma \delta }{2d}.
\label{k_min}
\end{equation}%
Therefore, the condition $h(\kappa )<0$ for some $\kappa \neq 0$ is 
\begin{equation*}
\frac{\left( d(\mu -3u_{0}^{2})-\gamma \delta \right) ^{2}}{4d}>\text{%
det }(A).
\end{equation*}%
A bifurcation occurs when $h_{\text{min}}=0$, this happens when the
condition 
\begin{equation*}
\text{det }(A)=\frac{\left( d\left( \mu -3u_{0}^{2}\right) -\gamma
\delta \right) ^{2}}{4d},
\end{equation*}%
is verified. This condition defines a critical diffusion $d_{c}$, which is
given as an appropriate root of%
\begin{equation*}
\left( \mu -3u_{0}^{2}\right) ^{2}d_{c}^{2}+2\left( -2\gamma +\mu \gamma
\delta -3\gamma \delta u_{0}^{2}\right) d_{c}+\gamma ^{2}\delta ^{2}=0.
\end{equation*}%
The model (\ref{Eq-TP_discrete-2}) for $d>d_{c}$ exhibits Turing
instability, while for $d<d_{c}$ it does not. Note that $d_{c}>1$. A
critical `wavenumber' is obtained using (\ref{k_min}) 
\begin{equation}
\kappa _{c}=\frac{d_{c}\left( \mu -3u_{0}^{2}\right) -\gamma \delta }{2d_{c}}%
=\sqrt{\frac{\text{det }(A)}{d_{c}}}.  \label{k_critical}
\end{equation}%
When $d>d_{c}$, there is a range of number of unstable positive waves $%
\kappa _{1}<\kappa <\kappa _{2}$, where $\kappa _{1}$, $\kappa _{2}$ are the
zeros of $h(\kappa )=0$, see (\ref{k1})-(\ref{k2}). We call to function $%
\rho (\kappa )$ the dispersion relation. We note that, within the unstable
range, $\text{\text{Re}}(\rho (\kappa ))>0$ has a maximum for the
wavenumber $\kappa _{min}^{(0)}$ obtained from (\ref{k_min}) with $d>d_{c}$.
Then as $t$ it increases, the behavior of $\left[ 
\begin{array}{c}
\left[ w_{L}^{\left( 1\right) }(I,t)\right] _{I\in G_{L}} \\ 
\\ 
\left[ w_{L}^{\left( 2\right) }(I,t)\right] _{I\in G_{L}}%
\end{array}%
\right] $ is controlled by the dominant mode, that is, those $e^{\rho
(\kappa )t}\left[ 
\begin{array}{c}
\boldsymbol{e}_{\kappa } \\ 
\boldsymbol{e}_{\kappa }%
\end{array}%
\right] $ with $\text{\text{Re}}(\rho (\kappa ))>0$, since the other
modes go to zero exponentially. We recall that $\kappa =\kappa \left(
L\right) $. For this reason, we use the notation $\kappa =\kappa _{L}$. Wit
this notation, 
\begin{equation}
w_{L}^{\left( j\right) }(I,t)\thicksim \sum_{\kappa _{1}<\kappa _{L}<\kappa
_{2}}C_{\kappa }\left( j,I\right) e^{\rho (\kappa _{L})t}\boldsymbol{e}%
_{\kappa }\text{, for }t\rightarrow \infty ,  \label{Turing_pattern}
\end{equation}%
where $j=1$, $2$.
\end{proof}

Digernes and his collaborators have studies extensively the problem of
approximation of spectra of Vladimirov\ operator $\boldsymbol{D}^{\alpha }$
by matrices of type $A_{L}^{\alpha }$, \cite{Digernes}-\cite{Digernes et al}%
. By using the fact that the eigenvalues $\varsigma \neq \lambda $ and
eigenfunctions $\Psi _{rnj}$ of ${\boldsymbol{D}_{0}^{\alpha }}$ are also
eigenvalues and eigenfunctions of $\boldsymbol{D}^{\alpha }$, and Theorem
4.1 in \cite{Bakken-Digernes}, one concludes that for $L$ sufficiently
large, the eigenvalues of matrix $A_{L}^{\alpha }$ approximate the
eigenvalues $\varsigma \neq \lambda $ of ${\boldsymbol{D}_{0}^{\alpha
}-\lambda }$, in a symbolic form $\Gamma _{L}\approx \Gamma \smallsetminus
\left\{ \lambda \right\} $.

\section{\label{Section_6} Numerical approximations of Turing patterns}

This section presents numerical approximations of Turing patterns associated
with specific $p$-adic FitzHugh-Nagumo systems. By suitable choosing of the
parameters ($\mu ,\gamma ,\delta ,\beta ,d$, with $d>1$), we find a region
where the conditions (1)-(5) of Theorem \ref{TheoremB} are satisfied. Then
we solve numerically the system of ODEs (\ref{Eq-TP_discrete-2}). Finally,
we give various visualizations of the solutions intending to show several
aspects of the Turing patterns. To construct a region (called the Turing
unstable region), we use an $\left( f_{u_{1}},g_{v_{1}}\right) $ plane,
i.e., we set%
\begin{equation*}
x=f_{u_{1}}=\mu -3u_{1}^{2}\text{, \ \ }y=g_{v_{1}}=-\gamma \delta .
\end{equation*}%
Figure \ref{Figure 2} shows a Turing unstable region associated with a
steady state\ of system (\ref{Eq-TP_discrete-2}). The parameters ($\mu
,\gamma ,\delta ,\beta ,d$, with $d>1$) that give rise to green points in
Figure \ref{Figure 2} correspond to some Turing pattern.

\begin{figure}[H]
	\centering
	\includegraphics[scale=0.8]{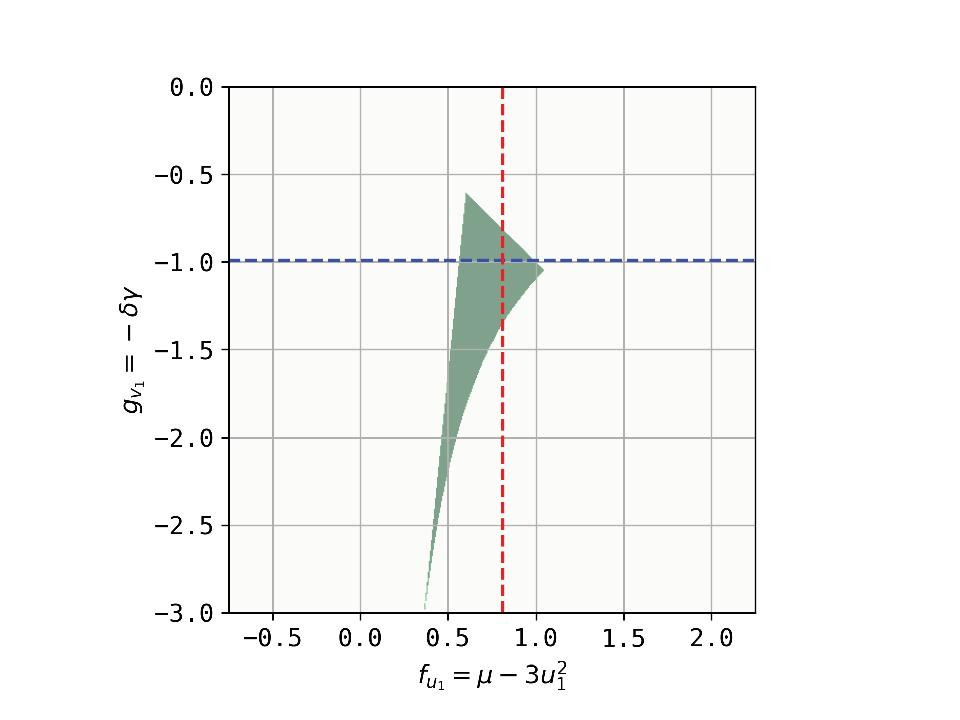}
	\caption{Turing unstable region. All the
		points in the green region of the $(f_{u_{1}}g_{v_{1}})$ -plane, which
		satisfies the conditions $(1)-(5)$ of Theorem \protect\ref{TheoremB}. The
		parameters are $p=2,\protect\mu =1.26,\protect\beta =0,\protect\delta =0.9, \protect\gamma =1.1,d=10,L=9$, and $(u_{1},v_{1})=(0.3858,0.4287)$.}
		\label{Figure 2}
\end{figure}

The last condition in Theorem \ \ref{TheoremB} is shown in the left part of
Figure \ref{Figure 3}. More precisely, the eigenvalues of matrix $%
A_{L}^{\alpha }$ between the dotted lines (which represent the values $%
\kappa _{1}$, $\kappa _{2}$) satisfy condition (6) in Theorem \ \ref%
{TheoremB}. The right part of Figure \ref{Figure 3} shows the eigenvalues of
operator $\boldsymbol{D}_{0}^{\alpha }-\lambda $, see Section \ref%
{Section_Spectrum}. For $L$ sufficiently large, the eigenvalues of $%
A_{L}^{\alpha }$ approach to the ones of $\boldsymbol{D}_{0}^{\alpha
}-\lambda $, such it was discussed at the end of Section \ref{Section_5}.


\begin{figure}[H]
	\centering
	\includegraphics[scale=0.6]{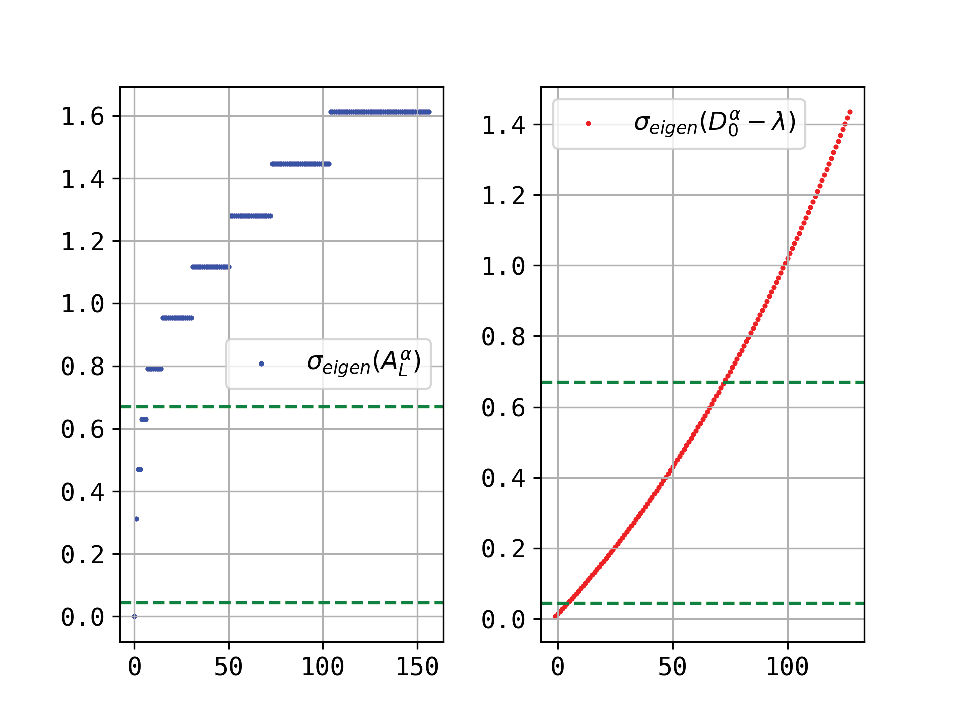}
	\caption{The left part of the figure
		shows the first $2^{9}$ eigenvalues of the matrix $A_{L}^{\protect\alpha }$,
		which is a discretization of the operator $D_{0}^{\protect\alpha }-\protect%
		\lambda $ The right part of the figure shows the first $2^{9}$ eigenvalues
		of $D_{0}^{\protect\alpha }-\protect\lambda $. Notice that eigenvalue $%
		\protect\lambda $ is very close to $1$. }
	\label{Figure 3}
\end{figure}

Figures \ref{Figure 4} and \ref{Figure 5} show the Turing patterns, which
are solutions of the Cauchy problem associated with system (\ref{Eq-TP_discrete-2}), with an initial datum close to $(u_{1},v_{1})$, for $t$
sufficiently large.

\begin{figure}[H]
	\centering
	\includegraphics[scale=0.6]{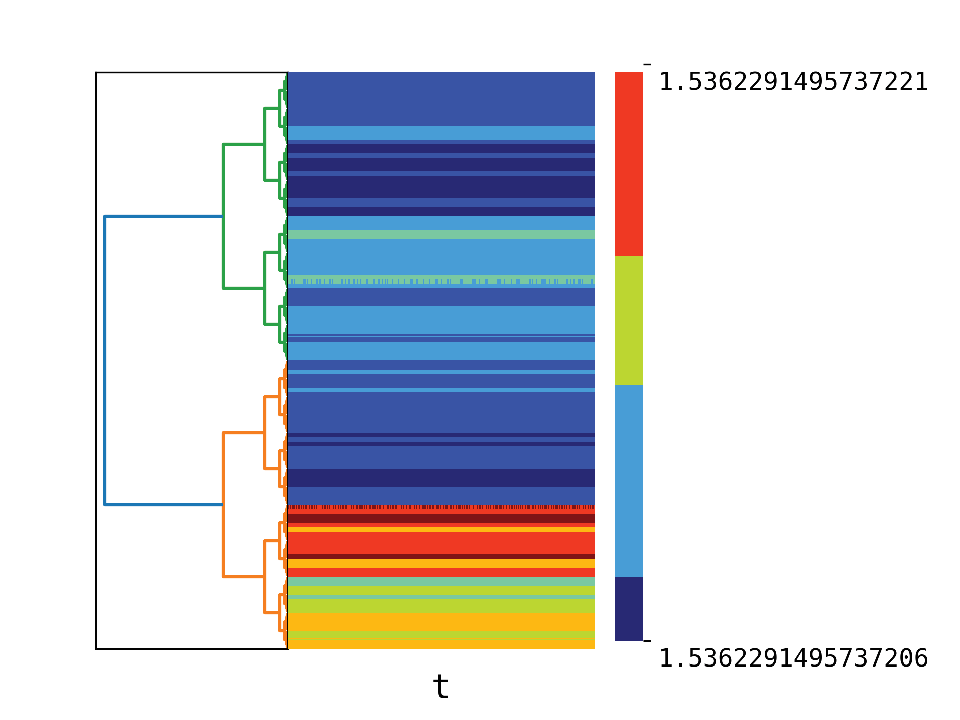}
	\caption{The activator states $ u_{L}(I,\cdot )$ for $731<t<2000$, and $L=9$. The vertical scale runs
		through the points of tree $G_{9}$.}
\label{Figure 4}
\end{figure}

\begin{figure}[H]
	\centering
	\includegraphics[scale=0.6]{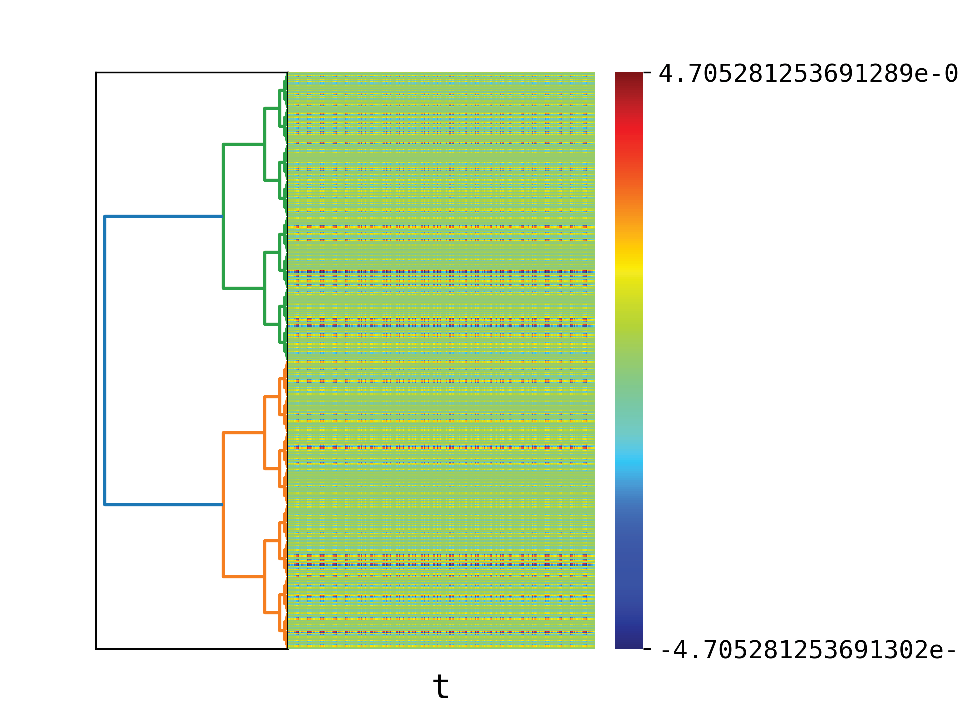}
	\caption{The inhibitor states $ 	v_{L}(I,\cdot )$ for $731<t<2000$, and $L=9$. The vertical scale runs
		through the points of tree $G_{9}$.}
\label{Figure 5}
\end{figure}

\begin{figure}[H]
	\centering
	\includegraphics[scale=0.6]{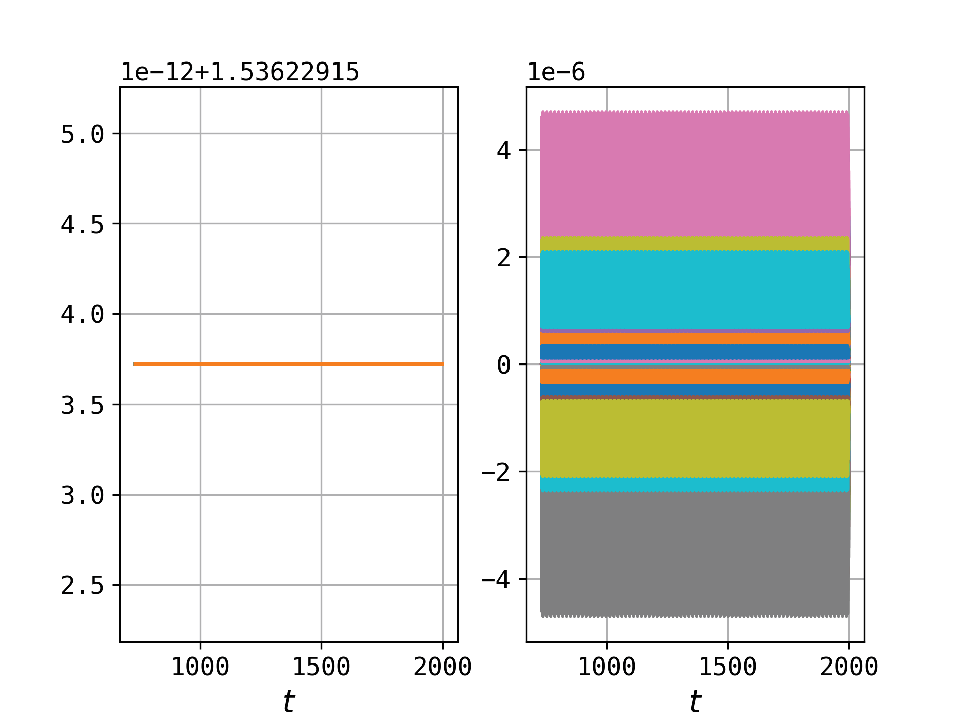}
	\caption{The left side of the figure
		shows the activator $u_{L}\left( I,\cdot \right) $, while the right side
		shows the inhibitor $v_{L}\left( I,\cdot \right) $. This figure shows the
		evolution of all the system states (\protect\ref{Eq-TP_discrete-2}) for time 
		$731<t<2000$. At time $t=0$, the initial datum for the Cauchy problem is $(%
		\protect\overset{\sim }{u_{1}},\protect\overset{\sim }{v_{1}})$, where $%
		\protect\overset{\sim }{u_{1}}$ is a sample of a Gaussian variable with mean 
		$u_{1}$ and variance $0.01$, and $\protect\overset{\sim }{v_{1}}$ is a
		sample of a Gaussian variable with mean $v_{1}$ and variance $0.01$. For any
		initial state $(\protect\overset{\sim }{u_{1}},\protect\overset{\sim }{v_{1}}%
		)$, the system (\protect\ref{Eq-TP_discrete-2}) develops the Turing pattern
		shown in this figure.}
	\label{Figure 6}
\end{figure}

\begin{figure}[H]
	\centering
	\includegraphics[scale=0.7]{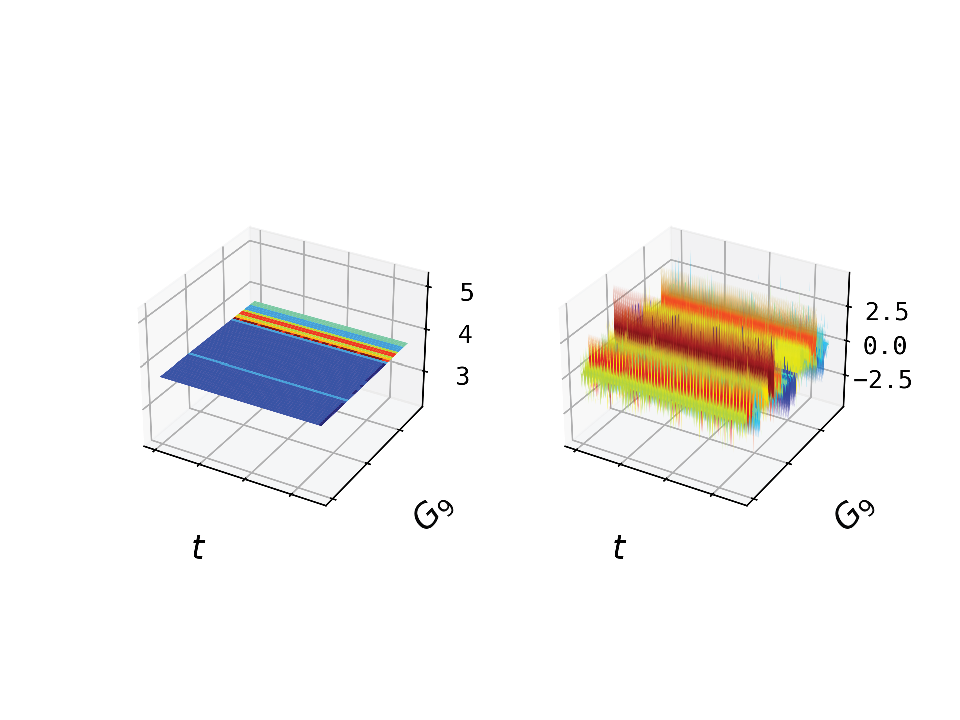}
	\caption{This Figure is a $3D$ version
		of Figure \protect\ref{Figure 6}. The left side of the figure shows the
		activator $u_{L}\left( I,\cdot \right) $, while the right side shows the
		inhibitor $v_{L}\left( I,\cdot \right) $. It shows the evolution of all the
		states of the system (\protect\ref{Eq-TP_discrete-2}) for time $731<t<2000$.
		The initial datum for the Cauchy problem is the same as in Figure 	\ref{Figure 6}. The Turing patterns are traveling waves.}
		\label{Figure 7}
\end{figure}

\end{document}